\documentclass[runningheads,orivec]{llncs}

\usepackage{complexity}
\usepackage[utf8]{inputenc}
\usepackage{cmap}
\usepackage[T2A]{fontenc}
\usepackage{csquotes}
\usepackage[english]{babel}
\usepackage{amssymb}
\usepackage{verbatim}
\usepackage[breaklinks]{hyperref}
\usepackage{algorithm}
\usepackage[noend]{algpseudocode}
\usepackage{comment}
\usepackage{bookmark}
\usepackage{mathtools}
\usepackage{makecell,xcolor}
\usepackage{tikz}
\usepackage{xspace}
\usepackage[strings]{underscore}

\DeclareMathOperator\DPLL{DPLL}
\newclass{\ONEDSPACE}{1DSPACE}
\newclass{\ONEDSPACEPRIME}{1DSPACE^\prime}

\DeclarePairedDelimiter\set{\lbrace}{\rbrace}
\DeclarePairedDelimiter\abs{\lvert}{\rvert}

\newcommand{\map}[2]{\set{#1 \mid #2}}
\newcommand{\BigO}{\mathcal{O}}
\newcommand{\SmallO}{o}

\newfunc{\fF}{F}
\newlang{\lL}{L}

\newcommand{\online}{online\xspace}
\newcommand{\offline}{offline\xspace}
\newcommand{\Online}{Online\xspace}

\newcommand{\offTme}{\offline Turing machine\xspace}

\newcommand{\onTme}{\online Turing machine\xspace}
\newcommand{\onTmes}{\online Turing machines\xspace}

\newcommand{\onTmeP}{\online Turing machine with shifted input\xspace}
\newcommand{\onTmesP}{\online Turing machines with shifted input\xspace}
\newcommand{\OnTmeP}{\Online Turing machine with shifted input\xspace}

\begin{document}

\title{Hard satisfiable formulas for DPLL algorithms using heuristics with small memory}

\titlerunning{Hard satisfiable formulas for DPLL with small memory}

\author{Nikita Gaevoy\inst{1,2}}

\authorrunning{N. Gaevoy}

\institute{St. Petersburg State University, \\
Universitetskaya nab., 7/9, St. Petersburg, Russia, 199034, \and
Steklov Institute of Mathematics at St. Petersburg, nab. r. Fontanki 27, \\
St. Petersburg, Russia, 191023 \\
\email{nikgaevoy@gmail.com}
}

\maketitle

\begin{abstract}
	$\DPLL$ algorithm for solving the Boolean satisfiability problem (SAT) can be represented in the form of a procedure that, using heuristics $A$ and $B$, select the variable $x$ from the input formula $\varphi$ and the value $b$ and runs recursively on the formulas $\varphi[x \coloneqq b]$ and $\varphi[x \coloneqq 1 - b]$. Exponential lower bounds on the running time of $\DPLL$ algorithms on unsatisfiable formulas follow from the lower bounds for tree-like resolution proofs. Lower bounds on satisfiable formulas are also known for some classes of $\DPLL$ algorithms such as ``myopic'' and ``drunken'' algorithms~\cite{Hirsch:2004}.

	All lower bounds are made for the classes of $\DPLL$ algorithms that limit heuristics access to the formula. In this paper we consider $\DPLL$ algorithms with heuristics that have unlimited access to the formula but use small memory. We show that for any pair of heuristics with small memory there exists a family of satisfiable formulas $\Phi_n$ such that a $\DPLL$ algorithm that uses these heuristics runs in exponential time on the formulas $\Phi_n$.
	
	\keywords{DPLL \and SAT \and {\onTmes} \and space-bounded computations \and sublinear space}
\end{abstract}

\section{Introduction}

$\DPLL$ (are named after the authors: Davis, Putnam, Logemann, Loveland~\cite{Davis:1960,Davis:1962}) algorithms are one of the most popular approaches to Boolean satisfiability problem (SAT). $\DPLL$ is an algorithm that takes the formula $\varphi$, uses heuristics $A$ and $B$ (which are the parameters of an algorithm) to choose a variable $x$ and the value $b$ that would be investigated first and makes recursive calls on formulas $\varphi[x \coloneqq b]$ and $\varphi[x \coloneqq \lnot b]$ if the first one is not satisfiable.

Every $\DPLL$ algorithm on any formula finds either its satisfying assignment or its tree-like resolution refutation. Therefore, exponential lower bounds for tree-like resolution (e.g. Tseitin formulae and its generalizations~\cite{Tseitin:1983,Urquhart:1987} and formulas based on pigeonhole principle~\cite{Haken:1985}) imply that any $\DPLL$ algorithm should work exponential time proving that corresponding formulas are unsatisfiable. However, the running time on satisfiable formulas may differ from the running time on unsatisfiable formulas and be even linear if heuristic $B$ is able to solve SAT. Moreover, satisfiable formulas are simpler for $\DPLL$-based SAT solvers that used on practice (and therefore more restrictive to the choice of its heuristics) rather than unsatisfiable ones.

Despite the fact that there is no hope to prove any nontrivial bounds for the running time of $\DPLL$ algorithms on satisfiable formulas with arbitrary polynomial time heuristics unless $\cP = \NP$, it is still interesting to prove lower bounds for $\DPLL$ algorithms that use heuristics from more narrow classes than $\cP$. Alekhnovich, Hirsch, and Itsykson~\cite{Hirsch:2004} proved exponential lower bounds on satisfiable formulas for two wide classes of $\DPLL$ algorithms: myopic $\DPLL$ and drunken $\DPLL$. Drunken $\DPLL$ has no restrictions on heuristic $A$, but the heuristic $B$ chooses its answer at random with equal probabilities. In myopic $\DPLL$ both heuristics has limited access on the input formula: they can read the whole formula with all negation signs erased and also they are able to read $n^{1 - \varepsilon}$ clauses precisely. Many formula simplification heuristics, such as elimination of a unit clauses can be simulated by myopic $\DPLL$, but others, such as subsumption heuristic (i.e. deletion of a clause that is a superset of an another clause) can not.

There are also a number of works concerning lower bounds for generalizations of $\DPLL$ algorithms. The paper~\cite{Itsykson:2011:Cut} gives lower bounds for $\DPLL$ algorithms with a cut heuristic, i.e. such additional heuristic $C$ that is able to decide not to make recursive calls of subformulas that it considered not ``perspective'' enough, and the paper~\cite{Itsykson:2017} gives lower bounds for $\DPLL(\parity)$ algorithms that can split not only by values of some variables, but also by values of its linear combinations. Papers~\cite{Cook:2009,Itsykson:2010,Itsykson:2011:Goldreich,Cook:2014} consider the generalization of $\DPLL$ designed to invert Goldreich's one-way function candidate and provide lower bounds on it.

\emph{Our contribution.} All discussed lower bounds are made for the classes of heuristics that limits access to the formula rather than computational power of a heuristic. In this work we consider ordinary $\DPLL$ with classes of deterministic heuristics that are bounded only by its space usage and not limited in its access to the formula, in particular $\DSPACE(\SmallO(\log))$. In order to prove an exponential lower bound for $\DPLL$ with heuristics from this class we use the notion of an \onTmes which can be considered as a formalization of streaming algorithms. Then we build an exponential reduction to its slight modification and prove an exponential lower bound for $\DPLL$ using \online heuristics with sublinear memory. Note that although the class of \online heuristics using sublinear memory appears to be relatively small, it is easy to see that it can express some formula simplification heuristics that can not be done using myopic or drunken algorithms, such as a subsumption of clauses that are not far away from each other (namely, at the distance $\BigO(\frac{n}{\polylog(n)})$ where $n$ is the size of the input) in the $k$-SAT formula.

\emph{Further research.} Our reduction of \offline algorithms to \online is specific to deterministic algorithms and cannot be straightforward generalized to randomized algorithms. It would be interesting to find a proper generalization of notion of an \online algorithm and a similar reduction to it for randomized algorithms. Also, despite the fact that almost all bounds presented in this work cannot be significantly improved without proving that $\cL \neq \NP$ there is an logarithmic gap between lower and upper bounds in the reduction of \offline algorithms to \online which is interesting to close.

\section{Sublinear space}

In order to work with classes of small memory we use the definition of a Turing machine with separate read-only input tape and read-write working tape. We also add separate write-only output tape for Turing machines whose output is more than one bit. By memory configuration of a Turing machine we mean a tuple of configuration of the working tape and the current state of the Turing machine.

\begin{definition}
	An \onTme is a Turing machine with the additional restriction that the input tape head can be shifted only in one direction.
\end{definition}

\begin{definition}
	$\ONEDSPACE(f)$ is the class of all languages recognized by \onTme using at most $f(n)$ memory.
\end{definition}

Obviously, any \onTme is also an \offTme, so $\ONEDSPACE(f) \subseteq \DSPACE(f)$. However, $\ONEDSPACE(f) = \DSPACE(f)$ for any $f = \Omega(n)$.

\begin{definition}
	\OnTmeP is an \onTme with a modification that the input string is shifted on the input tape to the size of the input. Thus, if the size of the input string is $n$, an \onTmeP must read $n$ empty symbols before it starts to read the input itself.
\end{definition}

\begin{definition}
	$\ONEDSPACEPRIME(f)$ is a class of all languages recognized by \onTmeP.
\end{definition}

This modification gives additional ability to read size of the input before reading input string itself. Obviously, $\ONEDSPACE(f) \subseteq \ONEDSPACEPRIME(f)$. Now we show that modified \onTmes are strictly more powerful than regular Turing machines even when the latters are restricted to a substantially larger space.

\begin{lemma}
	$\ONEDSPACEPRIME(\log n) \setminus \ONEDSPACE(o(n))$ is not empty.
\end{lemma}
\begin{proof}
	Consider language $\lL$ consisting of binary strings $s$ such that the binary representation of $\abs s$ is a prefix of $s$. It is easy to see that $\lL \in \ONEDSPACEPRIME(\log)$. Consider arbitrary \onTme $M$ that recognizes $\lL$. We show that after reading $k$ symbols $M$ should use at least $C k$ cells on working tape for some input where $C$ is some constant depending only on $M$. Assume the opposite, then there are two distinct words $s$ and $t$ such that $\abs s = \abs t = k$, their first symbol is $1$ and $M$ moves to the same configuration after reading $s$ and $t$. Consider words $S = s \cdot 1^{[s] - k}$ and $T = t \cdot 1^{[s] - k}$, where $[s]$ means number which binary representation is $s$. Configuration of $M$ after reading $s$ and $t$ are the same, so $M$ accepts $S$ if and only if $M$ accepts $T$, but $S \in \lL$ when $T \notin \lL$ which leads us to contradiction.
\end{proof}

However, there is no difference when the space is very low.

\begin{lemma}[\cite{Stearns:1965}]\label{SmallMemoryREG}
	$\ONEDSPACE(o(\log)) = \REG$
\end{lemma}

\begin{lemma}
	$\ONEDSPACEPRIME(o(\log)) = \REG$.
\end{lemma}
\begin{proof}
	Consider arbitrary \onTmeP $M$. We show that $M$ either uses $\Omega(\log)$ memory or recognizes a regular language. We split the work of $M$ into two phases. During the first phase it reads empty symbols before the start of the input and on the second phase it reads the input itself. Consider the graph of the memory configurations which $M$ passes when its input consists of infinite number of empty cells (informally, during the infinite first phase). Every vertex of this graph (i.e. configuration of $M$) has exactly one outgoing edge, since all symbols on the input tape are identical. Therefore, this graph may be either an infinite simple path or a cycle. In the case when this graph is a path $M$ should use $\Omega(\log)$ cells of memory after the first phase since no to memory configuration can appear twice. In the other case $M$ appears only in constant number of memory configurations during (and therefore after) the first phase. If $M$ uses at least $\Omega(\log)$ memory we already done, otherwise we can construct an \onTme $M_0$ which invokes the same computation as $M$ does on the second phase with all possible results of the first phase simultaneously and chooses one of them on the end of computation. Therefore, $M_0$ uses the same amount of memory as $M$ up to multiplicative constant and recognizes the same language as $M$, so by Lemma~\ref{SmallMemoryREG} $M$ either uses $\Omega(\log)$ memory or recognizes a regular language.
\end{proof}

\begin{definition}
	$\DSPACE(f, g)$ is a class of all functions computable by \onTme using $\BigO(f)$ cells of working tape and $\BigO(g)$ cells of output tape.

	Similarly, we define $\ONEDSPACE(f, g)$ and $\ONEDSPACEPRIME(f, g)$ for \onTmes and \onTmesP respectively.
\end{definition}

Note that $\DSPACE(f) = \DSPACE(f, 1)$ by definition. From now on we will consider only Turing machines with at least $\Omega(\log n)$ memory.

\begin{lemma}
	For any function $f(n) = \Omega(\log n)$ such that $f(n) \leq \frac{n}{2}$ for all sufficiently large $n$, if the function $1^n \mapsto f(n)$ belongs to $\DSPACE(\log f, \log f)$, then there exists language $\lL$ such that $\lL \in \DSPACE(\log f) \cap \ONEDSPACEPRIME(f)$ and $\lL \notin \ONEDSPACEPRIME(o(f))$.
\end{lemma}
\begin{proof}
	Let $\lL$ be the set of all $f$-periodic strings. More formally $\lL = \map{s}{\abs{s} = n \Rightarrow (\forall i < n - f(n) : s[i] = s[i + f(n)])}$.

	First we show that $\lL \in DSPACE(\log f)$. Consider \offTme that first computes the value $f(n)$, then rewinds the input tape to the start and checks equality in all pairs of symbols on distance of $f(n)$ in the input tape. It is easy to see that both parts can be done using only $\BigO(\log f)$ memory.

	Consider the computation of $f(n)$ as a function $1^n \mapsto f(n)$ on \offTme. In order to simulate this computation it suffices to store the configuration of the working tape and the position of the reading head on the input tape since all the symbols in the input tape are the same. It takes $\BigO(\log n)$ memory to store position on the input tape and $\BigO(\log f)$ memory to store configuration of working tape which is equal to $\BigO(\log n)$ in total since $f(n) \leq \frac{n}{2}$ for all sufficiently large $n$. Recall that $f = \Omega(\log n)$. Note that $f(n)$ can also be computed on an \onTmeP with $\BigO(f)$ memory by counting the number of empty symbols and simulation of the computation of function $f(n)$ on an \offTme.

	Any \online algorithm that recognizes $\lL$ should be in different configurations after reading different prefixes of the input of size $f(n)$, so $\lL \notin \ONEDSPACEPRIME(o(f))$. On the other hand, it suffices to store only the last $f(n)$ symbols of the input tape to recognize $\lL$ on an \onTme, therefore, $\lL \in \ONEDSPACEPRIME(f)$.
\end{proof}

\begin{theorem}\label{LargeOutputOffline2Online}
	For any function $f(n) = \Omega(\log\log n)$, if function $\fF$ can be computed on an \offTme $M$ using $f$ cells of working tape and binary working alphabet and $f$ can be computed on an \offTme as a function $1^n \mapsto f(n)$ using $\BigO(f \cdot 2^f \cdot \log \fF)$ memory, then $\fF \in \ONEDSPACEPRIME(f \cdot 2^f \cdot \log \fF, \log \fF)$.
\end{theorem}
\begin{proof}
	Consider an \offTme $M$ that computes $F$. Without loss of generality we can assume that $M$ moves working tape head on every step and stops only when it has its input tape head on the end of the input string. By position of a Turing machine we mean position of its head on the input tape.

	We construct an \onTme $M^\prime$ that computes $\fF$. At the start $M^\prime$ reads the size of the input $n$ and finds the value of the function $f(n)$ by simulation of computation of the function $1^n \mapsto f(n)$. Then, $M^\prime$ reads its input and simulates the work of $M$. We show how to do this explicitly.

	Let $M$ be in the memory configuration $x$ at position $k$. Consider path $\rho$ that $M$ will traverse over pairs of its memory configuration and the position of the input tape head, starting from the current position until it reaches position $k + 1$ or some halting configuration. Note that $\rho$ implicitly depends on the input and can be infinite. Let $h_k(x)$ be the function that returns memory configuration at the end of $\rho$ and string that $M$ prints on the output tape in $\rho$. If $\rho$ is infinite or the string of all printed symbols is longer than $\log \fF$, i.e. the longest possible output, then $h_k(x)$ returns a special loop marker.

	Note that in order to simulate the work of $M$ it is enough for $M^\prime$ only to maintain $h_k(x)$ for all $x$ and current position $k$ and memory configuration in which $M$ first comes into position $k$. $h_0$ can be computed trivially. We show that for every $k$ function $h_{k}$ can be computed using only $h_{k - 1}$ and the $k$-th input symbol.

	Consider some fixed $k$ and memory configuration $x$. Recall that $M$ moves the input tape head on every step of its computation. Then, $M$ either moves head to the right and reaches position $k + 1$ or moves head to the left and then using $h_{k - 1}$ we can compute in which memory configuration $M$ will reach position $k$ next time and the string that $M$ will print to the output until it reaches this position. Let $g_k(x)$ be the function that computes these two values. Note that $h_k(x)$ is either equal to $g^i_k(x)$ for some $i$ or returns the loop marker. In order to recognize a cycle in $g^i_k(x)$ we compute $g^i_k(x)$ and $g^{2i}_k(x)$ in parallel and then if at some step of this computation the string printed in $g^{2i}(x)$ becomes too long or configurations computed by these two functions becomes equal, we mark $h_{k}(x)$ with the loop marker, otherwise the computation halts.

	In order to compute the value of $h_{k}(x)$ we need $\BigO(f \cdot \log \abs Q \cdot \log \fF)$ memory where $Q$ denotes the set of finite states used by $M$. Let $q$ be the number of memory configurations of $M$. Then $M^\prime$ uses $\BigO(\log q + q \log q \cdot \log \fF)$ cells of memory and hence $\fF \in \ONEDSPACEPRIME(f \cdot 2^f \cdot \log \fF, \log \fF)$.
\end{proof}

\begin{corollary}\label{Offline2Online}
	For any function $f(n) = \Omega(\log\log(n))$, if language $\lL$ can be computed on \offTme $M$ using $f$ cells of working tape and binary working alphabet and $f$ can be computed on an \offTme as a function $1^n \mapsto f(n)$ using $\BigO(f \cdot 2^f)$ memory, then $\lL \in \ONEDSPACEPRIME(f \cdot 2^f)$.
\end{corollary}
\begin{proof}
	Immediately follows from the fact that any language is a function with only one bit output.
\end{proof}

\section{DPLL}

Consider $\DPLL_{A, B}$ algorithm for deciding the satisfiability of CNF-formula $\varphi$ parametrized by two heuristics $A$ and $B$. Heuristic $A$ takes formula $\varphi$, chooses some variable in $\varphi$ and returns its number. Heuristic $B$ takes $\varphi$ and the number of a variable (chosen by $A$) and returns a value for this variable.

For purposes of our proof from now on we will consider the $\DPLL_H$ algorithm, which has a single heuristic for both choosing a variable and a value for it.

\begin{algorithm}
	\caption{$\DPLL_H$}

	\begin{algorithmic}[1]
		\Procedure{$\DPLL_H$}{$\varphi$} \Comment{$\varphi$ --- formula in CNF}
			\If{$\varphi$ is empty}
				\State \Return $satisfiable$
			\EndIf

			\If{$\varphi$ contains empty clause}
				\State \Return $unsatisfiable$
			\EndIf

			\State $(x, b) \gets H(\varphi)$ \Comment{$H$ choose both variable and its value}

			\If{$\DPLL_H(\varphi[x = b]) = satisfiable$}
				\State \Return $satisfiable$
			\EndIf

			\State \Return $\DPLL_H(\varphi[x = \lnot b])$
		\EndProcedure
	\end{algorithmic}

\end{algorithm}

The following lemma shows that $\DPLL_{A, B}$ and $\DPLL_H$ are almost equivalent in terms of space complexity of its heuristics.

\begin{lemma}\label{PrimeNotPrime}
	Let $A$ and $B$ use at most $f(n)$ memory on all inputs of size $n$. Then there is a heuristic $H$, such that $H$ uses at most $f(n) + \BigO(\log\log n)$ memory and $\DPLL_H$ makes the same recursive calls as $\DPLL_{A, B}$.
\end{lemma}
\begin{proof}
	Let $S$ be the string returned by $A$. Since $\abs{S} = \BigO(\log n)$, memory $f(n) + \BigO(\log\log n)$ suffices to compute the $k$-th symbol of string $S$. Consider algorithm $H$ which emulates algorithm $B$ and computes symbols of string $S$ every time $B$ access them using additional $f(n) + \BigO(\log\log n)$ memory.
\end{proof}

\section{Lower bounds for the running time of DPLL algorithms on satisfiable formulae}

We need the construction of boundary expander matrices from~\cite{Hirsch:2004}.

\begin{definition}[{\cite[Definition~2.1]{Hirsch:2004}}]
	Let $A$ be Boolean matrix. For a set of rows $I$ of $A$, its boundary $\partial I$ is a set of all columns such that there exists exactly one row in $I$ that contains $1$ on the intersection. $A$ is an $(r, s, c)$-boundary expander if

	\begin{enumerate}
		\item Every row of $A$ has at most $s$ ones.

		\item For any set of rows $I$ if $\abs I \leq r$, then $\abs{\partial I} \geq c \cdot \abs I$
	\end{enumerate}
\end{definition}

The formula $\Phi_{A, \vec{b}}$ encoding the system of linear equations $A\vec{x} = \vec{b}$ is the formula constructed as follows. For each row of the matrix $A$ we construct CNF-formula $\Phi_{A, \vec{b}, i}(x)$, encoding $(\bigoplus_{j \in S_i} x_j = b[i])$, where $i$ is the number of row, and $S_i$ is the set of column numbers in which the row with the number $i$ contains ones. We take $\Phi_ {A, \vec {b}} (x) \coloneqq \bigwedge_i \Phi_{A, b, i}(x)$. Note that the resulting formula has a conjunctive normal form.
We identify the system of linear equations with the formula encoding it.

\begin{lemma}[{\cite[Theorem~3.1, Lemma~2.1, Remark~3.1]{Hirsch:2004}}]\label{Expander}
	There exists a family of Boolean matrices $(A_n)$ such that for every $n$,
	\begin{enumerate}
		\item $A_n$ has size $n \times n$.

		\item $A_n$ is a full rank matrix.

		\item Every row in $A_n$ has exactly three ones.

		\item Every column in $A_n$ has $\BigO(\log n)$ ones.

		\item $A_n$ is an $(\frac{n}{\log^{14} n}, 3, \frac{11}{13})$-boundary expander.
	\end{enumerate}
\end{lemma}

\begin{lemma}[{\cite[Lemma~3.7, Lemma~3.8]{Hirsch:2004}} {\cite[Corollary~3.4]{Ben-Sasson:2001}}]\label{RefutationSize}
	For any matrix $A$ which is an $(r, 3, c)$-boundary expander and any vector $b \notin Im(A)$ size of any tree-like resolution refutation of the system $A\vec{x} = \vec{b}$ must be at least $2^{\frac{cr}{2} - 3}$.
\end{lemma}

\begin{definition}
	A subformula is called elementary if it is obtained from the original formula by substituting a single variable.
\end{definition}

\newcommand{\midC}{\frac{18}{65}}
\newcommand{\CPrime}{\frac{9}{65}}
\newcommand{\makeC}[1]{\frac{\abs{\partial #1}}{\abs{#1}}}
\newcommand{\makeCParam}[2]{\frac{\abs{\partial_{#2} #1}}{\abs{#1}}}
\newcommand{\makeUNSATBound}{2^{\frac{9 r}{130} - 3}}

Let $A$ be a matrix of size $n \times n$ satisfying the conditions of Lemma~\ref{Expander}. Consider family of functions $f_{i, j}(x) = x_i \parity x_j$.

\begin{definition}
	Consider a CNF-formula $\Phi_{i, j; b}$ encoding the system of a linear equations $A\vec{x} = \vec{b}$ with additional equation $f_{i, j}(x) = 0$. We call a pair of indices $(i, j)$ bad if $i < j$ and for some value of $b$ there exists elementary unsatisfiable subformula of formula $\Phi_{i, j; b}$ such that the size of its tree-like resolution refutation is less than $\makeUNSATBound$, where $r$ is parameter of $A$.
\end{definition}

\begin{lemma}\label{MakeUNSAT}
	Let $A$ be a matrix of size $n \times n$ satisfying the conditions of Lemma~\ref{Expander}. There exist at most $\BigO(n \log^2 n)$ bad pairs of indices.\footnote{This bound is not tight, but we need only $\SmallO(n^2)$.}
\end{lemma}
\begin{proof}
	$\Phi_{i, j; b}$ encodes some system of a linear equations $B\vec{x} = \vec{d}$ where matrix $B$ depends on parameters $i, j$. We show that if $\abs {\partial_B I} \geq \max(2, \midC \abs{I})$ holds for any set $I$ of rows of the matrix $B$, then the pair $(i, j)$ is good. Consider an arbitrary elementary unsatisfiable subformula of $\Phi_{i, j; b}$. Let $B^\prime$ be the matrix obtained by removing the column that corresponds to a substituted variable. We show that $B^\prime$ is an $(r, 3, \CPrime)$-boundary expander.

We identify rows of the matrix $B^\prime$ with corresponding rows of the matrix $B$. Let $c^\prime \coloneqq \min_{I} \makeCParam{I}{B^\prime}$ be the parameter $c$ of $B^\prime$.

\begin{equation*}
	c^\prime = \min_{I} \makeCParam{I}{B^\prime} \geq \min_{I} \frac{\abs{\partial_{B} I} - 1}{\abs{I}} \geq \min_{I} \frac{\max(2, \midC \abs{I}) - 1}{\abs{I}}
\end{equation*}

Let function $f(t)$ be $\frac{\max(2, \midC t) - 1}{t}$. It attains its minimum of $\CPrime$ at $t = \frac{65}{9}$, so $c^\prime \geq \CPrime$ and $B^\prime$ is an $(r, 3, \CPrime)$-boundary expander. By Lemma~\ref{RefutationSize} the size of the minimal tree-like resolution refutation of the considered subformula is at least $\makeUNSATBound$ and the pair $(i, j)$ is good.

Consider a bad pair of indices $(i, j)$. This pair corresponds to some set of rows $I$ of the matrix $A$ such that the condition $\abs{\partial_B I^\prime} < \max(2, \midC \abs {I^\prime})$ holds for the set of rows $I^\prime$ obtained from $I$ by adding the row encoding $f_{i, j} = 0$. It is easy to see that $\abs{\partial_B I^\prime} \geq \abs{\partial_B I} - 2$. The condition $\abs{\partial_B I^\prime} < 2$ does not hold for $\abs{I} \geq 4$ because parameter $c$ of the matrix $A$ equals $\frac{11}{13}$, which is greater than $\frac34$. Then $\makeCParam{I^\prime}{B} \geq \frac{\abs{\partial_B I} - 2}{\abs{I} + 1} \geq \frac{c \abs{I} - 2}{\abs{I} + 1} = c - \frac{2 + c}{\abs{I} + 1}$. But $c - \frac{2 + c}{\abs{I} + 1} \geq \midC$ when $\abs{I} \geq 4$, therefore $\abs{I} \leq 3$.

Now for each $k \leq 3$ we bound the amount of bad pairs corresponding to the sets of rows of the size $k$.

Each set of the size $1$ corresponds to exactly three bad pairs. Consider a set of the size $2$. If $\abs{\partial_B I} \geq 4$, then it cannot correspond to any bad pair because $\makeCParam{I^\prime}{B} \geq \frac23 > \midC$. Otherwise this set consists of two rows that have two common variables. Each set of this type corresponds to exactly one bad pair, but there are at most $\BigO(n \log n)$ such sets because every column of the matrix $A$ contains $\BigO(\log n)$ ones and therefore every row can be in at most $\BigO(\log n)$ pairs.

The only remaining case is $\abs{I} = 3$. If $\abs{\partial_B I} \geq 4$ then $I$ can not correspond to any bad pai because $\abs{\partial_B I^\prime} \geq \abs{\partial_B I} - 2$, and $\makeCParam{I^\prime}{B} \geq \frac12 > \midC$. Moreover $\abs{\partial_B I} \geq 3$ because $c \abs{I} > 2$. Therefore, $I$ can correspond to some bad pair only of $\abs{\partial_B I} = 3$ and in this case it corresponds to exactly three bad pairs. Note that $\abs{\partial_B I} = 3$ can be only if each row in $I$ has at least one common variable with some of the remaining two. But then there is a row that has common variables with both other rows and there are at most $\BigO(n \log^2 n)$ sets of this type because every column of the matrix $A$ contains $\BigO(\log n)$ ones.

\end{proof}

\begin{definition}
	A family of unsatisfiable formulae $(\Phi_k)$ is a family of hard unsatisfiable formulae, if the size of a minimal tree-like resolution refutation of any formula from $(\Phi_k)$ is at least $2^{\frac{\ell}{\log^{c} \ell}}$ for some constant $c$, where $\ell = \Omega(\abs{\Phi_k})$.
\end{definition}

\begin{definition}
	A family of satisfiable formulae $(\Phi_k)$ is a family of hard satisfiable formulae, if the family of all its elementary unsatisfiable formulae is hard.
\end{definition}

\begin{definition}
	Two Boolean vectors are called opposite if their sum is equal to a vector of all ones.
\end{definition}

\begin{lemma}\label{Commute}
	Let $A$ be the full rank Boolean matrix. If $A$ contains an odd number of ones in each row, then operators $A$, $A^{-1}$ and the addition of a unit vector commute.
\end{lemma}
\begin{proof}
	It suffices to prove that $A(\vec{x} \parity \vec{1}) = (A\vec{x}) \parity \vec{1}$ for all $x$. There is an odd number of ones in each row, so $A\vec{1} = \vec{1}$. Then

	\begin{equation*}
		A(\vec{x} \parity \vec{1}) = A\vec{x} \parity A\vec{1} = (A\vec{x}) \parity \vec{1}
	\end{equation*}
\end{proof}

\begin{theorem}\label{Main}
	For any \online algorithm $H$ using $\SmallO(\frac{n}{\log n})$ memory there exists a family of pairs of hard satisfiable formulae $(\Phi_m^1, \Phi_m^2)$ satisfying the following conditions

	\begin{enumerate}
		\item $\Phi_m^j$ is a formula over $\BigO(m)$ variables.

		\item Each formula has exactly one satisfying assignment and formulae in one pair have opposite satisfying assignments.

		\item $H$ returns the same answers for formulae in one pair.
	\end{enumerate}
\end{theorem}
\begin{proof}
	Without loss of generality we can assume that $H$ prints its answer only after reading the entire input. Fix $m$.
	Let $A$ be the Boolean matrix of size $m \times m$ from Lemma~\ref{Expander}. We will construct a formula with the number of literals linear in $m$. The bit size of the resulting formula will be $\BigO(m \log m)$, so $H$ will use at most $\SmallO(\frac{m \log m}{\log(m \log m)}) = \SmallO(m)$ memory. From now on we will identify the size of the formula with the number of literals in it.

	We create a variable for each column and each ``1'' in matrix $A$. Let $x_{i, j}$ denote the variable that corresponds to the ``1'' in the $i$-th row and the $j$-th column and $x_j$ denote the variable that corresponds to the $j$-th column. Note that we have a linear number of $x_{i, j}$ variables since each row of $A$ contains exactly three ones.

	Consider an arbitrary Boolean vector $\vec{q}$ and the $i$-th row of matrix $A$. Let $a, b$ and $c$ be the column numbers corresponding to ones in the $i$-th row of $A$. Consider formula $\varphi_{q, i} \coloneqq (x_{i, a} \parity x_{i, b} \parity x_{i, c} = q[i])$. Note that CNF representation of $\varphi_{q, i}$ consists of $\BigO(1)$ literals.
	We define the formula $\varphi_q \coloneqq \bigwedge_i \varphi_{q, i}$.

	Now consider the formula $\Phi_{q, w, d; a, b} \coloneqq (\varphi_q \lor u) \land (\varphi_w \lor \lnot u) \land \bigwedge_{i, j}(x_{i, j} = x_j \parity d[i]) \land \psi_{a, b}$ where $\psi_{a, b} \coloneqq x_a \parity x_b$. $\varphi_q$ and $\varphi_w$ are formulae of size linear in $m$, so the size of the formulae $(\varphi_q \lor u)$ and $(\varphi_w \lor \lnot u)$ after conversion to CNF will also be linear, and therefore the formula $\Phi_{q, w, d; a, b}$ will have linear size.

	Consider the formula $\xi_{q, d} \coloneqq \varphi_q \land \bigwedge_{i, j} (x_{i, j} = x_j \parity d[i])$. The values of the variables $x_{i, j}$ are uniquely determined by $x_j$ according to the formula $x_{i, j} = x_j \parity d[i]$, therefore $\xi_{q, d}$ is true if and only if the condition $\bigoplus_{j \in I} (x_j \parity d[i]) = q[i]$ is satisfied for all rows of the matrix $A$, where $I$ denotes the set of columns containing one in the $i$-th row. Every row of $A$ contains exactly three ones, so this condition is equivalent to $\bigoplus_{j \in I} x_j = q[i] \parity d[i]$. The conjunction of these conditions encodes the system $A\vec{x} = \vec{q} \parity \vec{d}$ which means that $\xi_{q, d}$ has unique satisfying assignment $A^{-1} (\vec{q} \parity \vec{d})$. Therefore, for a fixed value of $u$, the formula $\Phi_{q, w, d; a, b}$ will have at most one satisfying assignment.

	The part of the formula depending on the parameter $\vec {w}$ begins later than the end of the part of the formula depending on $\vec{q}$ and ends before the part of the formula depending on $\vec {d}$, therefore $H$ will read the formula parameters in the order $q, w, d$.
	Let $Q$ be the largest set of vectors $q$ that $H$ cannot distinguish, and $W$ be the largest set of vectors $w \in Q \parity \vec{1}$ that $H$ cannot distinguish under assumption that $q \in Q$. It is easy to see that $\abs{Q} \geq 2^{m - \SmallO(m)}$, and therefore $\abs{W} \geq 2^{m - \SmallO(m)}$.
	Let $\Phi_{q, w; a, b} \coloneqq \Phi_{q, w, d; a, b}$ where arbitrary element of $W \parity \vec{1}$ is selected as $\vec{d}$.
	Consider $\widetilde{W} \coloneqq W \parity \vec{d}$ and $\widetilde{Q} \coloneqq \widetilde{W} \parity \vec{1}$. Note that $\vec{0} \in \widetilde{Q}$ (and therefore $\vec{1} \in \widetilde{W}$).

	By Lemma~\ref{Commute}, $A$ and the unit vector addition commute and $\widetilde{Q} = \widetilde{W} \parity \vec{1}$, therefore $A^{-1}\widetilde{Q} = (A^{-1}\widetilde{W}) \parity \vec{1}$.

	We choose $a_0, b_0$ in order that $\psi_{a_0, b_0}$ is not constant on the set $A^{-1}\widetilde{Q}$. We show that this can be done. We construct an equivalence relation on the coordinates of the space $\set{0, 1}^m$ (i.e. bits) as follows. $i \sim j$ if and only if $q[i] \parity q[j]$ is constant on all $q \in A^{-1}\widetilde{Q}$. There are at least $m - \SmallO(m)$ equivalence classes since $\abs{A^{-1}\widetilde{Q}} \geq 2^{m - \SmallO(m)}$. Therefore there exist $\Omega(m^2)$ functions that are not constants on $A^{-1}\widetilde{Q}$, and by Lemma~\ref{MakeUNSAT} there is also such $\psi_{a_0, b_0}$ that the size of the refutation of any unsatisfiable elementary subformula of the formula $(A\vec{x} = \vec{q}) \land \psi_{a_0, b_0}$ is exponential.

	Note that the chosen $\psi_{a_0, b_0}$ will be also non constant on the set $A^{-1}\widetilde{W}$, since these sets consist of opposite elements. Now we choose $q_0 \in \widetilde{Q}$ and $w_0 \in \widetilde{W}$ such that $\psi_{a_0, b_0}(A^{-1}\vec{q_0}) \neq \psi_{a_0, b_0}(\vec{0}) = \psi_{a_0, b_0}(A^{-1}\vec{0})$ and $\psi_{a_0, b_0}(A^{-1}\vec{w_0}) \neq \psi_{a_0, b_0}(\vec{1})$. Recall that $A$ contains exactly three ones in each row, which means that $A\vec{1} = \vec{1}$ and therefore $\psi_{a_0, b_0}(A^{-1}\vec{w_0}) \neq \psi_{a_0, b_0}(A^{-1}\vec{1})$. Consider formulae $\Phi_{\vec{0} \parity \vec{d}, \vec{w_0} \parity \vec{d}; a_0, b_0}$ and $\Phi_{\vec{q_0} \parity \vec{d}, \vec{1} \parity \vec{d}; a_0, b_0}$. Note that the corresponding formula parameters are indistinguishable for $H$ by construction, which means that $H$ answers the same on both formulae.

	We show that both formulae have exactly one satisfying assignment and their satisfying assignments are opposite.
	The formula $\Phi_{q, w; a, b}$ with $\psi_{a_0, b_0}$ removed, it has exactly one satisfying assignment for each value of $u$ ($A^{-1}(\vec{q} \parity \vec{d})$ for $u = 0$ and $A^{-1}(\vec{w} \parity \vec{d})$ for $u = 1$). For the considered formulae, $q$ takes the values $\vec{0} \parity \vec{d}$ and $q_0 \parity \vec{d}$ (and $w$, respectively, $\vec{w_0} \parity \vec{d}$ and $\vec{1} \parity \vec{d}$), on which $\psi_{a_0, b_0}$ takes different values by construction.
	Moreover, $\psi_{a_0, b_0}$ takes value $1$ on assignments corresponding to parameters $q = \vec{0} \parity \vec{d}$ and $w = \vec{1} \parity \vec{d}$, which correspond to opposite values of $u$ and are themselves opposite, since $A$ commutes with the unit vector addition by Lemma~\ref{Commute}.

	It remains to show that the constructed family of a formulae is a family of hard formulae.
	After substitution of any variable, if the formula becomes unsatisfiable, then the resulting formula has an unsatisfiable subformula of the form $\varphi_q \land \bigwedge_{i, j} (x_{i, j} = x_j \parity d[i]) \land \psi_{a, b}$, possibly with one substituted variable, the size of the refutation of which is not less than the size of the refutation of an elementary unsatisfiable subformula of a formula of the form $(A\vec{x} = \vec{q}) \land \psi_{a, b}$ which is exponential.
\end{proof}

\begin{corollary}
	For any \online heuristic $H$ that uses $\SmallO(\frac{n}{\log n})$ memory, there exists a family of satisfiable formulae such that $\DPLL_H$ makes at least $2^{\frac{\ell}{\log^c \ell}}$ recursive calls on formulae from this family for some $c$, where $\ell = \Omega(\frac{n}{\log n})$.
\end{corollary}
\begin{proof}
	Consider the family of pairs of formulae from Theorem~\ref{Main} and in each pair choose the formula on which $H$ goes into an unsatisfiable subformula after the first step. The size of the minimal tree-like resolution refutation of this subformula is exponential, so $\DPLL_H$ runs exponential time on it.
\end{proof}

\begin{corollary}
	For any \offline heuristic $H$ that uses $(1 - \varepsilon) \log n$ cells of memory over the binary alphabet for some positive $\varepsilon$, there exists a family of satisfiable formulae such that $\DPLL_H$ makes at least $2^{\frac{\ell}{\log^c \ell}}$ recursive calls on formulae from this family for some $c$, where $\ell = \Omega(\frac{n}{\log n})$.
\end{corollary}
\begin{proof}
	By Theorem~\ref{LargeOutputOffline2Online}, there exists an equivalent \online heuristic $H^\prime$ that uses $\BigO(\log n \cdot (1 - \varepsilon) \log n \cdot 2^{(1 - \varepsilon) \log n})$ memory. It is easy to see that $\log n \cdot (1 - \varepsilon) \log n \cdot 2^{(1 - \varepsilon) \log n} = \SmallO(\frac{n}{\log n})$. Thus, $H^\prime \in \ONEDSPACEPRIME(\SmallO(\frac{n}{\log n}), \log)$.
\end{proof}

\section*{Acknowledgments}

The author is grateful to Alexander Okhotin for helpful discussions and also grateful to Edward~A.~Hirsch, who supervised this work.

\bibliographystyle{splncs04}
\bibliography{biblio}

\end{document}